\documentclass[12pt,orivec]{llncs}
\usepackage[a4paper,hmargin=2.5cm,vmargin=3cm]{geometry}
\usepackage[english]{babel}
\usepackage{amssymb,amsmath,amsfonts}
\usepackage{float}

\usepackage[latin1]{inputenc}

\usepackage{ifthen}

\newboolean{SUBMISSION}
\newboolean{rework}
\title{One Cyclic Codes over $\mathbb{F}_{p^k}+v\mathbb{F}_{p^k}+v^2\mathbb{F}_{p^k}+...+v^r\mathbb{F}_{p^k}$}
\setboolean{SUBMISSION}{false}
\setboolean{rework}{true}
\begin{document}

\pagenumbering{arabic}
\pagestyle {plain}

\ifthenelse{\boolean{SUBMISSION}}{
	\author{}
	\institute{
	\vspace{17ex}
}
}{
	\pagestyle{headings}
	\author{
	    Ousmane Ndiaye
	}
		\institute{
Universit\'e Cheikh Anta Diop de Dakar, FST, DMI, LACGAA,\\ Senegal,\\
\email{ousmane3.ndiaye@ucad.edu.sn}
	}
}

\maketitle
\begin{abstract}
In this paper, we investigate cyclic code over the ring $\mathbb{F}_{p^k}+v\mathbb{F}_{p^k}+v^2\mathbb{F}_{p^k}+...+v^r\mathbb{F}_{p^k}$, where $v^{r+1}=v$, $p$ a prime number, $r>1$ and $\gcd(r,p)=1$, we prove as generalisation of \cite{Sol15} that these codes are principally generated, give generator polynomial and idempotent depending on idempotents over this ring as response to an open problem related in \cite{qian05}. we also give a gray map and proprieties of the related dual code.
\medskip

\end{abstract}

\section{Introduction} \label{sec:introduction}
Linear code over finite rings receives intensive study  in the last decades of twentieth. This study  over finite rings was motivated after the establishment of gray maps between non-linear codes and linear code over $\mathbb{Z}_4$ e.g. in \cite{HCSS94} and \cite{PQ96}.
Calderbank et al. give the structure of cyclic codes over $\mathbb{Z}_{p^{m}}$ in \cite{Cal95}.\\
C. Bachoc introduced the linear codes $\mathbb{F}_p+u\mathbb{F}_p$. These codes was also studied and them self-dual over $\mathbb{F}_2+v\mathbb{F}_2$ and $\mathbb{F}_2+v\mathbb{F}_2+v^2\mathbb{F}_2$ in [\cite{Bon99},\cite{Dou97},\cite{Uda99}, \cite{Shi13}].\\
Recently some cyclic codes are studied on chain rings over $\mathbb{F}_p$ where $p$ is prime. In \cite{Yas09} Y. Yasemin establish a relation between cyclic code $\mathbb{F}_{p^k}+v\mathbb{F}_{p^k}+v^2\mathbb{F}_{p^k}$ and code over $\mathbb{F}_p^k$, and investigate the structure of cyclic codes over the
ring $\mathbb{F}_3+v\mathbb{F}_3$ in \cite{Yas10}. In \cite{Sol15} cyclic code over $\mathbb{F}_p+v\mathbb{F}_p+v^2\mathbb{F}_p$ was studies in term of there idempotents and skew codes over $\mathbb{F}_q+v\mathbb{F}_q+v^2\mathbb{F}_q$ in \cite{Ashraf15}, by given for all of them a gray map on a field.\\
\par Some generalizations have been done by on codes over $\mathbb{F}_p+v\mathbb{F}_{p}+...+v^r\mathbb{F}_p$, in \cite{qian05}, J. Qian et al. introduced cyclic codes on those rings by analogy with $\mathbb{Z}_{p^m}$ case. They also set an open problem by suggesting an investigation on idempotents generator of those rings and determining the existence of self-dual code.\\
In this paper, we investigate cyclic code over the ring $\mathbb{F}_{p^k}+v\mathbb{F}_{p^k}+v^2\mathbb{F}_{p^k}+...+v^r\mathbb{F}_{p^k}$, where $v^{r+1}=v$, $p$ a prime number, $r>1$ and $\gcd(r,p)=1$, we prove as generalisation of \cite{Sol15} that these code are principally generated, give generator polynomial and idempotent depending on idempotents over $R_r$. we also give a gray map and proprieties of the related dual code.\\

The material is organized as follows. The next section contains the basics of codes over
rings that we need for further notice. Section 3 derives the structure of cyclic codes over R,
Section 4 derives the structure of dual cyclic codes over R and we conclude in section 5.

\section{Preliminaries} \label{sec:prelim}

Let $A$ be a commutative ring, a code $C$ of length $n$ over $A$ is an $A$-submodule of $A^n$.The Hamming weight of a codeword is the number of non-zero components.
 an element $e$ of $A$ is called idempotent if $e^2=e$.
Let $x =(x_1,x_2,...,x_n)$ and $y =(y_1,y_2,...,y_n)$ be two elements of $A^n$, the
inner product of $x$ and $y \in R_{r}^n$ is defined by,
\begin{center}
$xy =\sum^{n}_{i=1}x_iy_i$
\end{center}
Therefore from a linear code $C$ of length $n$ over $A$ , the corresponding dual code is define by
$C^{\perp}=\{ x \in R^{n} | x.c=0 ~\forall c \in C\}$. $C$ is self-dual if $C = C^{\perp}$, $C$ is
self-orthogonal if $C \subseteq C^{\perp}$. Two codes are equivalent if one can be obtained from the other by permuting the
coordinates.\\

A cyclic code $C$ of length n over $A$ is a linear code with property that if $c =(c_0,...,c_{n-1}) \in C$ then $\sigma(c)=(c_{n-1},c_0,...,c_{n-2}) \in C$. We assume that p is not divisible by n, and we represent codeword by polynomials. Then cyclic codes are ideals of the ring $A[x]/(x^n-1)$.\\
On every ring with at least one idempotent, one may apply the Peirce decomposition given by the following lemma.
\begin{lemma}
If $e$ is an idempotent in a ring $A$ not necessary commutative and not necessary unitary, then one can decompose $A$ as the direct sum of four components, each of them related to $e$. Concretely,\\
 \begin{center}
 $A=eAe \oplus eA(1-e) \oplus (1-e)Ae \oplus (1-e)A(1-e)$
 \end{center}
 This decomposition is known as the Peirce decomposition of $A$ with respect to $e$.
 \emph{If $A$ have not identity $(1-e)A$ means}
\begin{center}
$\{ x-ex | x \in A\}$.
\end{center}
\end{lemma}
In the commutative case we have
\begin{center}
$A=eA \oplus (1-e)A$
\end{center}
\begin{lemma}
Let  $e_1, ..., e_n$ be orthogonal nonzero idempotents of a commutative Ring $A$ with the property that $1 = \sum_{i=1}^{n} e_{i}$, then \begin{center}
$A=e_1A\oplus  ...\oplus  e_nA$
\end{center}
\end{lemma}
Let $R_{r}=\mathbb{F}_{p^k}+v\mathbb{F}_{p^k}+v^2\mathbb{F}_{p^k}+...+v^r\mathbb{F}_{p^k}$, where $v^{r+1}=v$, $p$ a prime number, $r>1$ and $\gcd(r,p)=1$.\\
This may be seen as $R_{r}=\mathbb{F}_{p^k}[v]/<v^{r+1}-v>$ all polynomial on $\mathbb{F}_{p^k}[v]$ of degree at most equal to $r$, where $\mathbb{F}_{p^k}$ is the primitive field. $R_{r}$ is a characteristic p ring of size $p^{r+1}$.
\begin{proposition}
  Let
\begin{center}
$e_{1}=\frac{1}{r}v+\frac{1}{r}v^2+...+\frac{1}{r}v^{r}$,
\end{center}

\begin{center}
 $e_{2}=-\frac{1}{r}v-\frac{1}{r}v^2-...-\frac{1}{r}v^{r-1}+\frac{r-1}{r}v^{r}$,
\end{center}
\begin{center}
$e_{3}=1-v^{r}$
\end{center}
 three elements of $\in R_r$ . Then $e_i$ are orthogonal nonzero idempotents verifying the Pierce conditions over $R_{r}$
\end{proposition}
\begin{proof}..\\
\begin{itemize}
  \item $e_1+e_2+e_3= 1?$
  \begin{eqnarray*}
    e_1+e_2+e_3 &=& \frac{1}{r}\sum_{i=1}^{r-1}v^{i}(1-1)+v^{r}(\frac{r-1}{r}+\frac{1}{r}-1)+1\\
          &=&  1
  \end{eqnarray*}
  \item $e_i^2=e_i?$
  \begin{eqnarray*}
    e_1^2 &=& \frac{1}{r^2}(v+v^2+...+v^{r})(v+v^2+...+v^{r})\\
          &=& \frac{1}{r^2}(v^2+v^3+v^4+...+v^{r+1}+\\
          & & ~~~~~~~~~~~v^3+v^4+...+v^{r+1}+v^{r+2}+\\
          & & ~~~~~~~~~~~~~~~+v^4+...+v^{r+1}+v^{r+2}+v^{r+3}+\\
          & & ~~~~~~~~~~~~~~~~~~~~~~~~~~~ .... +\\
          & & ~~~~~~~~~~~~~~~~~~~~~+v^r+v^{r+1}+v^{r+2}+v^{r+3}+...+v^{2r-1}+\\
          & & ~~~~~~~~~~~~~~~~~~~~~~~~~~~+ v^{r+1}+v^{r+2}+v^{r+3}+...+v^{2r-1}+v^{2r})\\
          &=& \frac{1}{r^2}(v^2+2v^3+...+jv^{j+1}+...+(r-1)v^r+\\
          & & rv^{r+1}+(r-1)v^{r+2}+...+(r-j)v^{r+j+1}+...+v^{2r})\\
          &=& \frac{1}{r^2}(~~~~~~~~~~~v^2+...+~~~~~~~~jv^{j+1}+...+(r-1)v^r+\\
          & & rv+(r-1)v^{2}+...+(r-j)v^{j+1}+...+~~~~~~~~v^{r})\\
          &=& \frac{1}{r^2}(rv+rv^2+rv^3+...+rv^r)\\
          &=& \frac{1}{r}(v+v^2+v^3+...+v^r)\\
          &=& e_1
  \end{eqnarray*}
  So $e_1$ is an idempotent element. By the way we show that $e_2^2=e_2$.
  \begin{eqnarray*}
    e_3^2 &=& (1-v^r)^2\\
          &=& 1-2v^{r}+v^{2r}\\
          &=& 1-2v^{r}+v^{r}\\
          &=& 1-v^{r}\\
          &=& e_3\\
  \end{eqnarray*}
  \item  $e_ie_j=0?$
  \begin{eqnarray*}
       e_ie_3 &=& e_i(1-v^{r})=e_i-e_iv^r=e_i-e_i=0\\
       e_1e_2 &=& \frac{-1}{r^2}(v+v^2+v^3+...+v^r)([v+v^2+v^3+...+v^r]-rv^r)\\
              &=& -[\frac{1}{r}(v+v^2+v^3+...+v^r)]^2+\frac{1}{r}v^r(v+v^2+v^3+...+v^r)\\
              &=& -e_1^2+e_1=0
  \end{eqnarray*}
\end{itemize}

 $\forall i,j \in  \{1, 2, 3\}$, we get $\sum_{i=1}^{3}e_{i}=1$, $e_{i}^{2}=e_{i}$ and $e_{i}e_{j}=0$,  then we get:
 \begin{center}
 $R_{r}=e_1R_{r}\oplus  e_2R_{r}\oplus  e_3R_{r}$.
 \end{center}
 \end{proof}
\par In the following, $R$ denotes the ring $e_1\mathbb{F}_{p^k}\oplus  e_2\mathbb{F}_{p^k}\oplus  e_3\mathbb{F}_{p^k}$ which is a sub-ring of $R_{r}$.
\begin{lemma}
    $e_{1}f_1+e_{2}f_2+e_{3}f_3$ is an idempotent in $R[x]/(x^n-1)$ if and only if $f_i$ are idempotents
in $\mathbb{F}_{p^k}[x]/(x^n-1)$; where i = 1, 2, 3.
\end{lemma}
\begin{proof}: \\
$\Rightarrow$: Let $e_{1}f_1+e_{2}f_2+e_{3}f_3$ be an idempotent element in $R[x]/(x^n-1)=e_1[\mathbb{F}_{p^k}[x]/(x^n-1)]\oplus  e_2[\mathbb{F}_{p^k}[x]/(x^n-1)]\oplus  e_3[\mathbb{F}_{p^k}[x]/(x^n-1)]$, \\
Then $(e_{1}f_1+e_{2}f_2+e_{3}f_3)^2=e_{1}f_1^2+e_{2}f_2^2+e_{3}f_3^2=e_{1}f_1+e_{2}f_2+e_{3}f_3$ that means $f_i^2=f_i, \forall i \in  \{1, 2, 3\}$.\\
$\Leftarrow$: When $f_i, \forall i \in  \{1, 2, 3\}$ are idempotents on $\mathbb{F}_{p^k}[x]/(x^n-1)$ it is clear that $e_{1}f_1+e_{2}f_2+e_{3}f_3$ is.
\end{proof}

\begin{definition}Gray map\\
Let $x =(x_1,x_2,...,x_n) \in R^n$ where $x_i =e_1s_i +e_2t_i+e_3u_i$, $i =1, 2,...,n$, the Gray map is given by,\\
\begin{center}
   $\phi : R^n \rightarrow \mathbb{F}_{p^k}^{3n}$\\
   $x \mapsto \phi(x)=(s(x),t(x),u(x))$
\end{center}
where $s(x)=(s_1,...,s_n)$, $t(x)=(t_1,...,t_n)$, and $u(x)=(u_1,...,u_n) \in \mathbb{F}_{p^k}^n$.\\
\end{definition}
If $A, B$ are codes over R, we write $A \oplus B$ to denote the code $\{a+b | a \in A, b\in B  \}$ and $A \otimes B$ to denote the code $\{(a,b) | a \in A, b\in B  \}$.
Let $C$ be a linear code over $R$, we define:
\begin{center}
  $C_1=\{s \in \mathbb{F}_{p^k}^n | \exists t, u \in  \mathbb{F}_{p^k}^n |e_1s +e_2t+e_3u \in C \}$
\end{center}
\begin{center}
  $C_2=\{t \in \mathbb{F}_{p^k}^n | \exists u, s \in  \mathbb{F}_{p^k}^n |e_1s +e_2t+e_3u \in C \}$
\end{center}
\begin{center}
  $C_3=\{u \in \mathbb{F}_{p^k}^n | \exists s, t \in  \mathbb{F}_{p^k}^n |e_1s +e_2t+e_3u \in C \}$
\end{center}
Following the same idea in [\cite{Yas10}, \cite{Sol15}], this means $C=e_1C_1\oplus e_2C_2 \oplus e_3C_3$, So any code over $R$ is characterized by its associative p-ary codes $C_1$, $C_2$, and $C_3$ and conversely.

\begin{theorem}\label{theocard}
Let $C$ be a linear code of length n over R, with its p-ary code $C_1$, $C_2$, $C_3$. Then $\phi(C)=C_1\otimes C_2\otimes C_3$ and $|C| = |C_1||C_2||C_3|$.
\end{theorem}
\begin{proof}: \\
$\Rightarrow$: Let $(s_1,...,s_n,t_1,...,t_n,u_1,...,u_n) \in \phi(C)$,  $ \exists x=(x_1,...,x_n) \in C$ such that $\phi(x)=(s_1,...,s_n,t_1,...,t_n,u_1,...,u_n)$. Since $\phi$ is injective, $x_i=e_1s_i+e_2t_i+e_3u_i$,\\
$x=(e_1s_1+e_2t_1+e_3u_1,...,e_1s_n+e_2t_n+e_3u_n)=e_1(s_1,...,s_n)+e_2(t_1,...,t_n)+e_3(u_1,...,u_n)$.\\
$(s_1,...,s_n) \in C_1, (t_1,...,t_n)\in C_2 ~and~ (u_1,...,u_n) \in C_3$,\\
So we have $(s_1,...,s_n,t_1,...,t_n,u_1,...,u_n) \in C_1\otimes C_2\otimes C_3$, then $\phi(C) \subseteq C_1\otimes C_2\otimes C_3$\\
$\Leftarrow$: Let $(s_1,...,s_n,t_1,...,t_n,u_1,...,u_n) \in C_1\otimes C_2\otimes C_3$ where $s=(s_1,...,s_n) \in C_1, t=(t_1,...,t_n)\in C_2 ~and~ u=(u_1,...,u_n) \in C_3$, Then there exists codewords $a=(a_1,..,a_n),b=(b_1,..,b_n),c=(c_1,..c_n) \in C$ such that:\\
$a=e_1s+e_2\lambda_1+e_3\lambda_2$, $b=e_1\lambda_3+e_2t+e_3\lambda_4$, $c=e_1\lambda_5+e_2\lambda_6+e_3u$ where $\lambda_i \in \mathbb{F}_{p^k}^n$, Since $C$ is linear , we have $h=e_1a+e_2b+e_3c=e_1s+e_2t+e_3u \in C$. Hence $(s_1,...,s_n,t_1,...,t_n,u_1,...,u_n)=\phi(h) \in \phi(C)$.\\
Therefore we have $\phi(C)=C_1\otimes C_2\otimes C_3$, and since $\phi$ is bijective it is easy to see $|C| = |C_1||C_2||C_3|$.
\end{proof}
\section{ Cyclic codes over}
\begin{lemma}\label{idemlem}
  Let $C$ be a cyclic code in a ring a commutative ring $A$, then:\\
\begin{enumerate}
  \item there exists a unique idempotent $e(x) \in C$ such that $C =<e(x)>$, and
  \item if $e(x)$ is a nonzero idempotent in $C$, then $C =<e(x)>$ if and only if $e(x)$ is a unity in $C$
\end{enumerate}
\end{lemma}
If $C_1$ and $C_2$ are codes of length n over a ring  $A$ , then $C_1 + C_2 =\{c_1 + c_2 | c_1\in C_1 and c_2\in C_2 \}$
is the sum of $C_1$ and $C_2$. Both the intersection and the sum of two cyclic codes are cyclic,
and their generator polynomials and generating idempotents are determined in the next theorem which generalise in \cite{Sol15}.

\begin{theorem}
Let $C_i$ be a cyclic code of length $n$ over $A$ with generator polynomial $g_i(x)$
and generating idempotent $e_i(x)$ for $i =1,2,...,t$. Then:
\begin{enumerate}
  \item[(i).] $\bigcap_{i=1}^{t} C_i$ has generator polynomial $lcm(g_1(x), ..., g_t(x))$ and generating idempotent $\prod_{i=1}^{t}e_i(x)$
  \item[(ii).] $\sum_{i=1}^{t} C_i$ has generator polynomial $\gcd(g_1(x), ..., g_t(x))$ and generating idempotent $\sum_{i=1}^{t} e_i(x)-\sum_{i<j}^{t} e_i(x)e_j(x)+\sum_{i<j<k}^{t} e_i(x)e_j(x)e_k(x)-...+(-1)^{t-1}\prod_{i=1}^{t}e_i(x)$.
\end{enumerate}
\end{theorem}
\begin{proof}: \\
  \begin{enumerate}
  \item[(i).] Let $g(x)=lcm(g_1(x), ..., g_t(x))$,  $C=\bigcap_{i=1}^{t} C_i$ and $c(x)$ the generator polynomial of $C$ .\\
  For $i =1,2,...,t$, $g_i(x)| g(x) \Rightarrow g(x) \in C_i$, then $g(x) \in C$ which means  $c(x)$ divides $g(x)$.\\
 On the other hand, for $i =1,2,...,t$, $c(x) \in C_i \Rightarrow g_i(x)| c(x)$, then $deg(c(x)) \geq deg(lcm(g_1(x), ..., g_t(x)))=deg(g(x))$.\\
      We get finally $c(x) | g(x)$ and $deg(c(x)) \geq  deg(g(x))$, then the generator polynomial  $c(x)$ of $C$ is $lcm(g_1(x), ..., g_t(x))$.\\
    Since $\prod_{i=1}^{t}e_i(x)$ is an idempotent element and unity in $C$, According to the lemma \ref{idemlem} it is a generating idempotent of $C$.
  \item[(ii).] Let $g(x) = \gcd(g_1(x), ..., g_t(x))$ and
\begin{center}
$e(x)=\sum_{i=1}^{t} e_i(x)-\sum_{i<j}^{t} e_i(x)e_j(x)+...+(-1)^{t-1}\prod_{i=1}^{t}e_i(x)$.
\end{center}
    $\forall i \in \{1,..,t\}$,  $g(x)|g_i(x)$, which shows that $C_i \subseteq <g(x)>$ implying $\sum_{i=1}^{t} C_i \subseteq <g(x)>$.\\
    On the other hand, It follows from the Euclidean Algorithm that
    $g(x) = g_1(x)a_1(x)+ ...+ g_t(x)a_t(x)$ for some $a_i(x) \in R[x] ~(i=1,...,t) $, that means $g(x) \in \sum_{i=1}^{t} C_i$ implying $<g(x)> \subseteq \sum_{i=1}^{t} C_i$. So $<g(x)> = \sum_{i=1}^{t} C_i$.\\
    By recurrence one can see that $e(x,2)$ is idempotent element, suppose it is now true until $t-1$ ie :\\
    $e(x, t)=\sum_{i=1}^{t-1} e_i(x)-\sum_{i<j}^{t-1} e_i(x)e_j(x)+\sum_{i<j<k}^{t-1}e_i(x)e_j(x)e_k(x)-...+(-1)^{t-2}\prod_{i=1}^{t-1}e_i(x)$ is idempotent element,
    So $(e(x,t) +e_t(x)-e_t(x)e(x,t))^2=(e(x,t) +e_t(x)-e_t(x)e(x,t))$.\\
    Hence  \begin{eqnarray*}
    e(x,t) +e_t(x)-e_t(x)e(x,t) &=& \sum_{i=1}^{t-1} e_i(x)-\sum_{i<j}^{t-1} e_i(x)e_j(x)+...+(-1)^{t-2}\prod_{i=1}^{t-1}e_i(x) \\
                                    & &+ e_t(x)- \sum_{i=1}^{t-1} e_i(x)e_t(x)+...+(-1)^{t-1}\prod_{i=1}^{t-1}e_i(x)e_t(x)\\
                                    &=&\sum_{i=1}^{t} e_i(x)-\sum_{i<j}^{t} e_i(x)e_j(x)+...+(-1)^{t-1}\prod_{i=1}^{t}e_i(x)
           \end{eqnarray*}
Let $c(x) \in \sum_{i=1}^{t} C_i$, then $c(x)=\sum_{i=1}^{t} c_i(x)$ where $c_i(x) \in C_i$. For each $i=1,...,t$, we may write the idempotent, under the isolation of one $e_i(x)$ as $e(x)=e(x,i) +e_i(x)-e_i(x)e(x,i)$
\begin{eqnarray*}
  c(x)e(x) &=& \sum_{i=1}^{t} c_i(x)e(x) \\
   &=&   \sum_{i=1}^{t} c_i(x)(e(x,i) +e_i(x)-e_i(x)e(x,i))\\
   &=&\sum_{i=1}^{t} (c_i(x)e(x,i) +c_i(x)e_i(x)-c_i(x)e_i(x)e(x,i))\\
   &=&\sum_{i=1}^{t} (c_i(x)e(x,i) +c_i(x)-c_i(x)e(x,i))\\
   &=& \sum_{i=1}^{t} c_i(x)=c(x)
\end{eqnarray*}
Since $e(x)$ is an idempotent element and unity in $\sum_{i=1}^{t} C_i$, According to the lemma \ref{idemlem} it is a generating idempotent of $\sum_{i=1}^{t} C_i$.

\end{enumerate}
\end{proof}
\begin{corollary}
  Let $C_i$ be a cyclic code of length n over $A$ for $i =1,2,...,t$. Then:
  \begin{eqnarray*}
    dim(C_1 +C_2 +...+ C_t) &=& \sum_{i=1}^{t}dim(C_i )-\sum_{i<j}^{t}dim(Ci \cap C j ) \\
    & & +\sum_{i<j<k}^{t}dim(C_i \cap C_j \cap C_k )-... \\
    & & +(-1)^{t-1}dim(C_1 \cap C_2 \cap...\cap C_t)
  \end{eqnarray*}
\end{corollary}
\begin{theorem}\label{thecyclic}
If $C=e_1C_1\oplus e_2C_2\oplus e_3C_3$ is a linear code of length $n$ over $R$, then
$C$ is a cyclic code over $R$ if and only if $C_1$,$C_2$,and $C_3$ are p-ary cyclic codes of length $n$.
\end{theorem}
\begin{proof}
  Let $\sigma$ be the shift operator. all cyclic codes are stable from $\sigma$.\\
One the one hand, let $C=e_1C_1\oplus e_2C_2\oplus e_3C_3$ be a linear code of length $n$ over $R$ and $C_i$ (i=1,...,n) p-ary  cyclic codes.
If $c=(c_1,...,c_n) \in C$ then , $c=(e_1s_1+e_2t_1+e_3u_1,...,e_1s_n+e_2t_n+e_3u_n)=e_1(s_1,...,s_n)+e_2(t_1,...,t_n)+e_3(u_1,...,u_n)$ with $(s_1,...,s_n) \in C_1, (t_1,...,t_n) \in C_2, (u_1,...,u_n) \in C_3$.\\
\begin{eqnarray*}
  \sigma(c) &=& \sigma(e_1(s_1,...,s_n)+e_2(t_1,...,t_n)+e_3(u_1,...,u_n)) \\
             &=& e_1\sigma(s_1,...,s_n)+e_2\sigma(t_1,...,t_n)+e_3\sigma(u_1,...,u_n) \\
\end{eqnarray*}
 Since $\sigma(s_1,...,s_n) \in C_1$, $\sigma(t_1,...,t_n)\in C_2$ and $\sigma(u_1,...,u_n) \in C_3$, then $\sigma(c) \in C$, which is in this case cyclic.\\
\par On the other hand, let $C=e_1C_1\oplus e_2C_2\oplus e_3C_3$ a cyclic code. For $s=(s_1,...,s_n) \in C_1$, $t=(t_1,...,t_n)\in C_2$ and $u=(u_1,...,u_n) \in C_3$, we have $e_1s +e_2t+e_3u \in C$ . Since $\sigma(e_1s +e_2t+e_3u)=(e_1\sigma(s) +e_2\sigma(t)+e_3\sigma(u)) \in C$, then $\sigma(s_1,...,s_n) \in C_1$, $\sigma(t_1,...,t_n)\in C_2$ and $\sigma(u_1,...,u_n) \in C_3$. Therefore $C_i$ are cyclic.
\end{proof}

\begin{corollary}
If $C$ is a cyclic code of length n over $R$, then $\phi(C)$ is a quasi-cyclic code of order 3 and length $3n$.
\end{corollary}
\begin{proof}
As $\phi(C)=C_1\otimes C_2\otimes C_3$ and $C$ cyclic. So according to the theorem  \ref{thecyclic}, $\phi(C)$ is concatenation of 3 cyclic codes of length $n$, hence it is quasi-cyclic code of order 3 and length $3n$.
\end{proof}

\begin{theorem}
If $C=e_1C_1\oplus e_2C_2\oplus e_3C_3$ is a cyclic code of length $n$
over $R$, then $C = <e_1g_1(x), e_2g_2(x),e_3g_3(x)>$ and $|C| =p^{3n-deg(g_1)-deg(g_2)-deg(g_3)}$,
where $g_1(x)$,$g_2(x)$, $g_3(x)$ are the generator polynomial of $C_1$,$C_2$,and $C_3$ respectively.
\end{theorem}
\begin{proof}
 Let $c(x) \in C \Leftrightarrow \exists a_1(x),a_2(x),a_3(x) \in R[x]$ s.t $c(x)=e_1g_1(x)a_1(x)+e_2g_2(x)a_2(x)+e_3g_3(x)a_3(x) \Leftrightarrow c(x) \in <e_1g_1(x), e_2g_2(x),e_3g_3(x)>$.\\
\begin{eqnarray*}
  |C|  &=& |\phi(C)| \\
   &=& |C_1||C_2||C_3| \\
   &=& p^{n-deg(g_1)}p^{n-deg(g_2)}p^{n-deg(g_3)}\\
   &=& p^{3n-deg(g_1)-deg(g_2)-deg(g_3)}
\end{eqnarray*}

\end{proof}
\begin{corollary}
For any cyclic code C of length $n$ over $R$ there is a unique polynomial $g(x)$ such that $H=<g(x)>$ and $g(x)|x^n-1$, where $g(x)=\frac{1}{r}\sum_{i=1}^{r-1}v^{i}(g_1(x)-g_2(x))+v^{r}(\frac{r-1}{r}g_1(x)+\frac{1}{r}g_2(x)-g_3(x))+g_3(x)$. Moreover if $g_1(x)= g_2(x)=g_3(x)$, then $g(x)=g_1(x)$.
\end{corollary}
\begin{corollary}
Let $C=e_1C_1\oplus e_2C_2\oplus e_3C_3$ be a cyclic code of length $n$ over $R$, where
$\gcd(n, p) = 1$, $C_i = <f_i>$, $g_i(i = 1, 2, 3)$ are idempotents, then there is only one idempotent
$e \in C$ such that $C =<e>$ where $e =e_1f_1(x)+e_2f_2(x)+e_3f_3(x)$.
\end{corollary}


\section{Cyclic Dual Codes}\label{sec:sec}
Recall that the ordinary inner product of vectors $x=(x_1,...x_n), y=(y_1,...,y_n)\in R^{n}$ is
$x.y=\sum_{i=1}^{n}x_iy_i$. Therefore from a linear code of length $n$ over $R$ , the corresponding dual code is define by
$C^{\perp}=\{ x \in R^{n} | x.c=0 \forall c \in C\}$. For a  polynomial $h$ of degree $k$, $h^{*}=x^{k}h(x^{-1})$ will denote its reciprocal polynomial .

\begin{lemma}
  The number of element in any nonzero linear code $C$ of length $n$  over $R_r$  is of the form $p^s$  Furthermore, the
dual code $C^{\perp}$ has $p^t$ codewords where $s+t=(k+1)n$
\end{lemma}
\begin{theorem}
Let $C_i$ be a cyclic code of length $n$ over $A$ with generator polynomial $g_i(x)$
and generating idempotent $e_i(x)$ for $i =1,2,...,t$. Then:\\

  $C_i^{\bot}$ has generator polynomial the reciprocal of $h_i(x)$ where $h_i(x)g_i(x)=x^n-1$ and generating idempotent $1-e_i(x^{-1})$

\end{theorem}
\begin{proof}: \\
 Since $e_i(x)(1-e_i(x))=0$, $e_i(x)$ is orthogonal to the idempotent $1-e_i(x^{-1})$, that is $1-e_i(x^{-1})\in C_i^{\bot}$.\\
    on the other hand, Let $c_i(x) \in C_i^{\bot}$, then $c_i^{\ast}(x)e(x)=0$ in $R[x]/<x^n-1>$.
    \begin{eqnarray*}
      c_i^{\ast}(x)(1-e_i(x))&=& c_i^{\ast}(x) \\
      \Rightarrow c_i^{\ast}(x^{-1})(1-e_i(x^{-1}))&=& c_i^{\ast}(x^{-1}) \\
      \Rightarrow x^{deg(c_i)}c_i^{\ast}(x^{-1})(1-e_i(x^{-1}))&=& x^{deg(c_i)}c_i^{\ast}(x^{-1}) \\
      \Rightarrow c_i(x)(1-e_i(x^{-1}))&=& c_i(x) \\
    \end{eqnarray*}
Since $1-e_i(x^{-1})$ is an idempotent element and unity in $C_i^{\bot}$, According to the lemma \ref{idemlem} it is a generating idempotent of $C_i^{\bot}$.

\end{proof}
\begin{theorem}
Let $C$ be a linear code of length $n$ over R, with its p-ary code $C_1$, $C_2$, $C_3$.
Then $C^{\bot}=e_1C_1^{\bot}\oplus e_2C_2^{\bot}\oplus e_3C_3^{\bot}$.
\end{theorem}
\begin{proof}: \\
Let $c \in e_1C_1^{\bot}\oplus e_2C_2^{\bot}\oplus e_3C_3^{\bot}$, then it exists $\tilde{c_i} \in C_i^{\bot}$ for $i= 1,2,3$ such that
$c=e_1\tilde{c_1}+ e_2\tilde{c_2}+ e_3\tilde{c_3}$. Let $x \in C$,  then it exists $c_i \in C_i$ for $i= 1,2,3$ such that
$x=e_1c_1+ e_2c_2+ e_3c_3$. Hence $x.c=0$ that means
\begin{center}
$e_1C_1^{\bot}\oplus e_2C_2^{\bot}\oplus e_3C_3^{\bot} \subseteq C^{\bot}$.
\end{center}
One the other hand, Let $x=e_1c_1+ e_2c_2+ e_3c_3 \in C=e_1C_1\oplus e_2C_2 \oplus e_3C_3$, and $c=e_1\tilde{c_1}+ e_2\tilde{c_2}+ e_3\tilde{c_3} \in C^{\bot}\subseteq R^n$, hence
$x.c=0=e_1c_1\tilde{c_1}+e_2c_2\tilde{c_2}+e_3c_3\tilde{c_3}$ $\Rightarrow$ :\\
$\left\{
  \begin{array}{ll}
    c_1\tilde{c_1}= & \hbox{0;} \\
    c_2\tilde{c_2}= & \hbox{0;} \\
    c_3\tilde{c_3}= & \hbox{0.}
  \end{array}
\right.$
Then $\tilde{c_i} \in C_i^{\bot}$ for $i= 1,2,3$ $\Rightarrow$ $c \in e_1C_1^{\bot}\oplus e_2C_2^{\bot}\oplus e_3C_3$.\\
Finally we have $C^{\bot}=e_1C_1^{\bot}\oplus e_2C_2^{\bot}\oplus e_3C_3^{\bot}$.
\end{proof}
\begin{corollary}
Let $C = <e_1g_1(x), e_2g_2(x),e_3g_3(x)>$ be a cyclic code of length $n$ over R, where
$\gcd(n, p) = 1$, $C_i = <g_i>$, $g_i(i = 1, 2, 3)$ are idempotents, then $C = <e_1g_1(x)+e_2g_2(x)+e_3g_3(x)>$ and the idempotent generator of $C^{\perp}$ is $1-e_1g_1(x^{-1})-e_2g_2(x^{-1})e_3g_3(x^{-1})$
\end{corollary}
\begin{corollary}
  $C$ is a cyclic self dual code of length $n$ over $R$ if and only if $C_1$,$C_2$,and $C_3$ p-ary cyclic self dual codes of length $n$.
\end{corollary}




\end{document}